\documentclass{elsarticle}  %for ILP
\usepackage{amssymb}
\usepackage{graphicx,amsmath}
\usepackage{times}
\usepackage{booktabs}
\usepackage{hyperref}
\usepackage{amsthm}
\usepackage{color}
\usepackage{tabularx}
\usepackage{multirow}
\usepackage{multicol}

\usepackage{lineno}

\newtheorem{theorem}{Theorem}
\newtheorem{lemma}[theorem]{Lemma}

\newtheorem{myrule}{Rule}

\newcommand{\comments}[1]{}
\newcommand{\problem}{{\sc{\{Claw, Diamond\}-Free Edge Deletion}}}
\newcommand{\probshort}{\sc{CDFED}}
\newcommand{\noins}{NO-instance}
\newcommand{\yesins}{YES-instance}
\newcommand{\cB}{{\mathcal{B}}}
\newcommand{\edge}[2]{(#1,#2)}
\newcommand{\nsmbg}{2k+2}
\newcommand{\nbg}{2k+3}

\newcommand{\bigo}{O}
\newcommand{\np}{NP}
\newcommand{\nph}{{\np}-hard}
\newcommand{\fpt}{FPT}
\newcommand{\fst}{$H$}%{\sf{fs}}
\newcommand{\myfig}[1]{Figure~\ref{#1}}
\newcommand{\abs}[1]{|#1|}

\newcommand{\EP}[4]{
\begin{center}
{\small
\begin{tabularx}{0.98\columnwidth}{ll}
\toprule
\multicolumn{2}{c}{\textsc{#1}} \\
\midrule
{\bf Input:}   & \parbox[t]{0.85\columnwidth}{#2\vspace*{1mm}}  \\
{\bf Parameter:}   & \parbox[t]{0.85\columnwidth}{#3\vspace*{1mm}}  \\
{\bf Question:}& \parbox[t]{0.85\columnwidth}{#4\vspace*{.5mm}} \\
\bottomrule
\end{tabularx}
}
\end{center}
}

\journal{Theoretical Computer Science}

\begin{document}
\begin{frontmatter}

\title{Improved Kernel and Algorithm for Claw and Diamond Free Edge Deletion Based on Refined Observations\tnoteref{mytitlenote}}
\tnotetext[mytitlenote]{This work was supported by the National Natural Science Foundation of China (61872048, 61972330), and the Postgraduate Scientific Research Innovation Project of Hunan Province (CX20200883).}

\author[wjl]{Wenjun Li}
\ead{lwjcsust@csust.edu.cn}

\author[wjl]{Huan Peng}
\ead{3898950557@qq.com}

\author[yyj]{Yongjie Yang
%\corref{cor1}
}
\ead{yyongjiecs@gmail.com}


\address[wjl]{School of Computer and Communication Engineering, Hunan Provincial Key Laboratory of Intelligent Processing of Big Data on Transportation,\\ Changsha University of Science and Technology, Changsha, China}
\address[yyj]{Chair of Economic Theory, Saarland University, Saarbr\"{u}cken, Germany}

\begin{abstract}
In the {\problem} problem ({\probshort}), we are given a graph~$G$ and an integer $k>0$, and the question is whether there are at most~$k$ edges whose deletion results in a graph without claws and diamonds as induced subgraphs. Based on some refined observations, we propose a kernel of~$\bigo(k^3)$ vertices and~$\bigo(k^4)$ edges, significantly improving the previous kernel of~$\bigo(k^{12})$ vertices and~$\bigo(k^{24})$ edges. In addition, we derive an $\bigo^*(3.792^k)$-time algorithm for {\probshort}.
\end{abstract}

\end{frontmatter}

%\linenumbers

\section{Introduction}
Graph modification problems consist in transforming a given graph into a desired graph by modifying the graph in a certain way (e.g., adding/deleting a limited number of vertices/edges). These problems are relevant to a wide range of real-world applications. As the number of modification operations allowed to be performed is expected to be small in many applications, graph modification problems have been extensively studied from the parameterized complexity perspective (see, e.g.,~\cite{
DBLP:journals/algorithmica/BetzlerBBNU14,
DBLP:journals/ipl/Cai96,
DBLP:journals/algorithmica/CyganMPPS14,
DBLP:conf/wg/KomusiewiczNN15,
DBLP:journals/tcs/LiuWYCC15,
DBLP:journals/algorithmica/MarxS12,
DBLP:journals/jda/Sritharan16}).
We refer to~\cite{DBLP:journals/corr/abs-2001-06867} for a comprehensive survey of the recent progress on the parameterized complexity of graph modification problems.

In this paper, we study the {\problem} problem ({\probshort}) which is a special case of the $\mathcal{H}$-{\sc{free edge deletion}} problem, where $\mathcal{H}$ is a set of graphs. A graph is $\mathcal{H}$-free if it does not contain any graph in~$\mathcal{H}$ as an induced subgraph. In the $\mathcal{H}$-{\sc{free edge deletion}} problem, we are given a graph~$G$ and an integer~$k$, and the question is whether there are at most~$k$ edges whose deletion results in an $\mathcal{H}$-free  graph. If~$\mathcal{H}$ consists of the claw and the diamond graphs, we have exactly the {\probshort} problem. A claw is a star with three leaves, and a diamond is a complete graph of four vertices with one arbitrary edge missing. Cygan~et~al.~\cite{DBLP:journals/mst/CyganPPLW17} initialized the study of {\probshort}. In particular, they proved that {\probshort} is {\nph} and does not admit a subexponential-time algorithm unless the Exponential Time Hypothesis (ETH) fails and, moreover, these hold even when restricted to graphs of maximum degree~$6$. On the positive side, they derived a kernel of~$\bigo(k^{12})$ vertices and~$\bigo(k^{24})$ edges for an annotated version of {\probshort}. In particular, in the annotated version, we are given an additional subset~$S$ of vertices and the question is whether we can delete at most~$k$ edges between vertices in~$S$ so that the resulting graph does not contain any claw or diamond as induced subgraphs. When~$S$ is the vertex set of the given graph, we have the {\probshort} problem.
Our main contribution is the following result.

\begin{theorem}
\label{thm-kernel}
{\probshort} admits a kernel of~$\bigo(k^3)$ vertices and~$\bigo(k^4)$ edges.
\end{theorem}

%As a claw has three edges and a diamond has five edges, a naive branching strategy leads to an $\bigo^*(5^k)$-time {\fpt}-algorithm.
Then, based on some refined observations we develop an $\bigo^*(3.792^k)$-time algorithm for {\probshort}%
%\footnote{{\bigo^*} is the notation {\bigo()} with the polynomial dependence being omitted}%
. \footnote{A recent 3-page paper  by Tsur, posted on https://arxiv.org, has improved our result to an algorithm running in~$\bigo^*(3.562^k)$ time~\protect\cite{DBLP:journals/corr/abs-1908-07318}. However, for the following reasons, we present our algorithm in the paper. First, the algorithm of Tsur is built upon the main idea of our algorithm. Particularly, in our algorithm, when branching a diamond, we consider one additional vertex outside the diamond. Tsur's algorithm refines our algorithm by further considering one more vertex.
Second, in the new algorithm, Tsur used a python program to compute the worst branching configuration (the program code has not been reported in~\protect\cite{DBLP:journals/corr/abs-1908-07318}). We provide the detailed description of all possible cases. Therefore, our algorithm is more transparent and self-contained.}

\begin{theorem}
\label{thm-algorithm}
{\probshort} can be solved in~$\bigo^*(3.792^k)$ time.
\end{theorem}

%\noindent{\bf{Related Works}.}
It should be pointed out that, when~$\mathcal{H}$ consists of only the diamond, the {$\mathcal{H}$-{\sc{Free Edge Deletion}} problem is known to admit a kernel of~$\bigo(k^3)$ vertices~\cite{%DBLP:conf/esa/00010SY18,
DBLP:conf/iwpec/SandeepS15}. However, for~$\mathcal{H}$ consisting of only the claw, whether the {$\mathcal{H}$-\sc{Free Edge Deletion}} problem admits a polynomial kernel remained open heretofore.
%It should be also pointed out that graphs without claws and diamonds as induced subgraphs are a subclass of line graphs. Concretely, they are exactly line graphs of triangle-free graphs~\cite{DBLP:journals/siamdm/MetelskyT03}. Whether the {\sc{Line Graph Edge Deletion}} problem has a polynomial kernel remains open. The {\sc{Line Graph Edge Deletion}} problem asks whether we can transform a given graph into a line graph by deleting~$k$ edges.

\section{Preliminaries}
The notation and terminology used in this paper mainly follow the work of Cygan~et~al.~\cite{DBLP:journals/mst/CyganPPLW17}. We study only undirected simple graphs.
%Unless stated otherwise, all numerical data in this paper are integer numbers.

For a graph~$G$, we use~$V(G)$ and~$E(G)$ to denote its vertex set and its edge set respectively.
%We consider only simple graphs, i.e., there is no loop on each vertex and between every two vertices there is at most one edge.
For a vertex $v\in V(G)$,~$N_G(v)$ is the set of all neighbors of~$v$, i.e., $N_G(v)=\{u \mid \edge{v}{u} \in E(G)\}$.
An {\it{isolated vertex}} is a vertex without any neighbor. Two vertices are {\it{adjacent}} if there is an edge between them. For an edge~$\edge{v}{u}$, we say that~$v$ (resp.~$u$) and $\edge{v}{u}$ are {\it{incident}}.

For $S\subseteq V$, the subgraph of~$G$ induced by~$S$, denoted~$G[S]$, is the graph with vertex set~$S$ and edge set $\{\edge{v}{u}\in E(G) \mid v, u\in S\}$. In addition,~$E_G(S)$ is the set of all edges between vertices in~$S$ in~$G$, i.e., $E_G(S)=E(G[S])$. Throughout this paper, we  write~$E(S)$ for~$E_G(S)$. For $A\subseteq E(G)$ (resp.\ $A\subseteq V(G)$), $G-A$ is the graph obtained from~$G$ by deleting all edges (resp.\  vertices) in~$A$.
For a set~$F$ of pairs over~$V(G)$ such that $F\cap E(G)=\emptyset$,~$G+F$ is the graph obtained from~$G$ by adding edges between every pair in~$F$.
A subset $S\subseteq V$ is a {\it{clique}} if there is an edge between every two vertices in~$S$.
A {\it{maximal clique}} is a clique that is not a proper subset of any other clique.

A graph~$G$ is {\it{isomorphic}} to another graph~$G'$ if there is a bijection $f: V(G)\rightarrow V(G')$ such that for every $v, u\in V(G)$, it holds that $\edge{v}{u}\in E(G)$ if and only if $\edge{f(v)}{f(u)}\in E(G')$.

A {\it{claw}} is a graph with four vertices~$c$,~$\ell_1$,~$\ell_2$, and~$\ell_3$, and three edges~$\edge{c}{\ell_1}$,~$\edge{c}{\ell_2}$, and~$\edge{c}{\ell_3}$. The vertex~$c$ (resp.\ each~$\ell_i$, $1\leq i\leq 3$) is called the {\it{center}} (resp.\ {\it{leaf}}) of the claw. A {\it{diamond}} is a graph with four vertices and five edges.
%A graph is {\emph{$\{\text{claw, diamond}\}$-free}} if it does not contain any claws or diamonds as induced subgraphs.

A graph is $\{\text{claw, diamond}\}$-free if it does not contain any claws or diamonds as induced subgraphs.
A subset~$S$ of edges in a graph~$G$ is called a CDH (claw and diamond hitting) set of~$G$ if $G-S$ is  $\{\text{claw, diamond}\}$-free.
The problem studied in this paper is formally defined as follows.
\EP
{{\problem} (\probshort)}
{A graph $G=(V, E)$ and a positive integer~$k$.}
{$k$.}
{Does~$G$ have a CDH set of size at most~$k$?}

\medskip

{\bf{Parameterized Complexity~\cite{DBLP:series/txcs/DowneyF13}.}} A {\it parameterized problem} is a subset $Q \subseteq \Sigma^* \times
\mathbb{N}$ for some finite alphabet~$\Sigma$, where the second component is called the {\it parameter}. A {\it{kernelization}} of a parameterized problem~$Q$ is an algorithm that transforms every instance $(x, k)$ of~$Q$ into an instance $(x', k')$ of~$Q$ such that
(1) the algorithm runs in polynomial time in the size of the instance $(x, k)$;
(2) $(x,k) \in Q$ if and only if $(x',k') \in Q$;
(3) $k'\leq f(k)$ for some computable function~$f$;
and (4) $|x'|\leq g(k)$ for some computable function~$g$.
The new instance $(x',k')$ is called a {\it{kernel}} of~$Q$.

\section{A Structural Property of \{Claw, Diamond\}-Free Graphs}
Before giving the kernelization, let us first recall some important properties of \{claw, diamond\}-free graphs, which have been studied in~\cite{DBLP:journals/mst/CyganPPLW17}. We need the following notations from~\cite{DBLP:journals/mst/CyganPPLW17}.
A {\it{simplicial}} vertex is a vertex whose neighbors form a clique. A {\it{bag}} is a maximal clique or a simplicial vertex. For a \{claw, diamond\}-free graph~$H$, let~$\cB(H)$ be the set of all bags of~$H$.

\begin{lemma}[\cite{DBLP:journals/mst/CyganPPLW17}]
\label{lem-bags-computable-polynomial-time}
For every \{claw, diamond\}-free graph~$H$,~$\cB(H)$ can be computed in polynomial time.
\end{lemma}

Lemma~\ref{lem-bags-computable-polynomial-time} directly implies that~$\cB(H)$ contains polynomially many bags.
The following lemma describes a structural property of \{claw, diamond\}-free graphs.

\begin{lemma}[\cite{DBLP:journals/mst/CyganPPLW17}]
\label{lem:bag-decomposition}
Let~$H$ be a \{claw, diamond\}-free graph without isolated vertices. Then, the following conditions hold:
\begin{enumerate}
	\item\label{p:ver} every vertex in~$H$ is included in exactly two bags;
	\item\label{p:edg} every edge in~$H$ is in exactly one bag;
	\item\label{p:inters} every two bags in~$\cB(H)$ share at most one common vertex; and
	\item\label{p:nonadj} for two bags~$B_1$ and~$B_2$ in $\cB(H)$ sharing a vertex~$v$, there are no edges between $B_1\setminus \{v\}$ and $B_2\setminus \{v\}$ in~$H$.
\end{enumerate}
\end{lemma}

In fact, in Lemma~\ref{lem:bag-decomposition}, the first condition prohibits the existence of induced claws, and other conditions prohibit the existence of induced diamonds. We remark that every nonisolated simplicial vertex belongs to the bag consisting of only itself and another bag of size at least two. {\myfig{fig-property-of-claw-diamond-free-graphs}} illustrates the above lemma. %For formal proof, we refer to~\cite{DBLP:journals/mst/CyganPPLW17}.
\begin{figure}[h!]
\begin{center}
\includegraphics[width=0.5\textwidth]{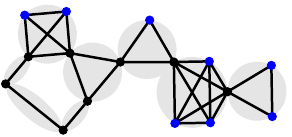}
\end{center}
\caption{A \{claw, diamond\}-free graph and its bags. Each bag of size at least two is emphasized in a gray area. Each blue vertex is a simplicial vertex.}
\label{fig-property-of-claw-diamond-free-graphs}
\end{figure}

\section{The Kernelization}
In this section, we study a kernelization of {\probshort} based on several reduction rules.
Let~$(G,k)$ be a given instance. As the kernelization in~\cite{DBLP:journals/mst/CyganPPLW17}, our kernelization begins with finding an arbitrary maximal collection (packing) of edge-disjoint induced claws and diamonds which clearly can be done in polynomial time. Let~$X$ denote the set of vertices that appear in some claw or diamond in the packing. Such a set~$X$ is called a {\it{modulator}} of~$G$. Clearly, $G-X$ is \{claw, diamond\}-free. If $|X|>4k$, we need to delete at least $k+1$ edges to destroy all induced claws and diamonds. So, in this case, the kernelization immediately returns a trivial {\noins}. We study five reduction rules to reduce the number of vertices in $G-X$.

A reduction rule is {\it{sound}} if each application of the reduction rule does not change the answer to the instance. An instance is {\it{irreducible}} with respect to a set of reduction rules if none of these reduction rules is applicable to the instance. We assume that, when a reduction rule is introduced, the instance is irreducible with respect to all reduction rules introduced before.
Moreover, after each application of a reduction rule, we recalculate a modulator~$X$ of~$G$.
%None of our reduction rules decreases the parameter~$k$.
In the proof of the soundness of a reduction rule, we will use $(G, k)$ and $(G', k)$ to denote the instances before and after the application of the reduction rule, respectively (none of our reduction rules changes the parameter~$k$).

%Our kernelization is based on five reduction rules given below. 
The first rule is trivial and the soundness of the rule is easy to see.

\begin{myrule}
\label{rule_remove_components}
If there are isolated vertices, delete all of them.
\end{myrule}

Now we introduce four new rules.
To apply these rules, we classify all bags into three sets.
Bags whose vertices are all adjacent to at least one vertex in common in the modulator are called attached bags, others that share some common vertices with attached bags are called border bags, and the remaining ones are called outlier bags. By the forbidden of induced claws, one shall see that for every vertex in~$X$ there can be at most three bags attached to it. This directly bounds the number of attached bags with respect to the size of the modulator. Then, starting from this, the four new reduction rules described below shrink these three types of bags iteratively in both their numbers and sizes.

Formally, a bag $B\in \cB(G-X)$ is {\it{attached}} to a vertex $x\in X$ if
\begin{enumerate}
\item either $|B|\geq 2$ and~$x$ is adjacent to all vertices in~$B$; or
\item $B=\{v\}$,~$x$ is adjacent to~$v$, but~$x$ is not adjacent to all vertices in the other bag including~$v$. See Figure~\ref{fig-attach}.
\end{enumerate}

\begin{figure}
\begin{center}
\includegraphics[width=0.65\textwidth]{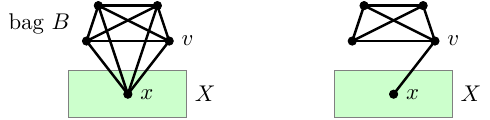}
\end{center}
\caption{On the left side,~$B$ is attached to~$x$, but the bag~$\{v\}$ is not attached to~$x$. In the right-handed figure the bag~$\{v\}$ is attached to~$x$.}
\label{fig-attach}
\end{figure}

A bag is {\it{attached}} if it is attached to at least one vertex in~$X$. For an unattached bag $B\in \cB(G-X)$ which shares a vertex with some attached bag, we call~$B$ a {\it{border bag}}. Note that a border bag can be also a simplicial vertex. For instance, in the graph on the left side of \myfig{fig-attach},~$\{v\}$ is a border bag%
%($\{v\}$ is not a border bag in the right-side graph)
. An unattached bag that is not a border bag is called an {\it{outlier}} bag.

%By the fact that each {\yesins} has a size-bounded modulator, and by some structural properties, we can bound the size of border bags. Using a similar strategy, given size-bounded border bags, we are able to bound outlier bags.

The following rule shrinks outlier bags.

\begin{myrule}
\label{rule_remove_edges_in_outlier_bags}
If there is an outlier bag~$B\in \cB(G-X)$, delete all edges between vertices in~$B$ from~$G$.
\end{myrule}

A bag is \emph{small} if it has less than~$\nsmbg$ vertices; and is {\it{big}} otherwise.

\begin{myrule}
\label{rule_deleting_edges_in_irrevalent_bags}
If there is a border bag~$B\in \cB(G-X)$ which is of size at least~$2$ and does not share any vertex with any small attached bags, then delete all edges between vertices in~$B$ from~$G$.
\end{myrule}

Before proving the soundness of Rules~\ref{rule_remove_edges_in_outlier_bags} and~\ref{rule_deleting_edges_in_irrevalent_bags}, we study some properties.
For a bag $B\in \cB(G-X)$, let $A(B)\subseteq X$ be the set of vertices in~$X$ to which~$B$ is attached. Observe that if~$B$ is of size at least~$2$, then $B\cup A(B)$ is a clique in~$G$.
This is true because otherwise there is an induced diamond in $G[B\cup A(B)]$ that is edge-disjoint from every induced claw and diamond in~$G[X]$, contradicting the maximality of~$X$.
In addition, observe that deleting one edge from a clique of size at least~$4$ results in several induced diamonds. Hence, if a clique is too large, deleting one edge from the clique triggers the deletions of many other edges, in order to destroy the induced diamonds. These observations lead to the following lemma.

\begin{lemma}[\cite{DBLP:journals/mst/CyganPPLW17}]
\label{lem:large_bag_disjiont_with_solution}
Let $B\in \cB(G-X)$ be a big bag. %Then, $B\cup A(B)$ is a clique in $G$. Moreover,
Then, every CDH set of~$G$ of size at most~$k$ does not include any edge in $E(B\cup A(B))$.
\end{lemma}

\begin{lemma}[\cite{DBLP:journals/mst/CyganPPLW17}]
\label{lem:diamond_inside_bag}
Let~$H$ be a subgraph (not necessarily induced) of~$G$ isomorphic to a diamond.
Let $B\in\cB(G-X)$ be an unattached bag containing at least two vertices of~$H$.
Then,~$B$ contains all vertices of~$H$.
\end{lemma}

One can observe that if the bag~$B$ contains exactly two (or three) vertices of~$H$ in Lemma~\ref{lem:diamond_inside_bag}, then at least one of the conditions in Lemma~\ref{lem:bag-decomposition} is violated. See Figure~\ref{fig-diamond-in-a-bag} for an illustration.

\begin{figure}
\begin{center}
\includegraphics[width=0.85\textwidth]{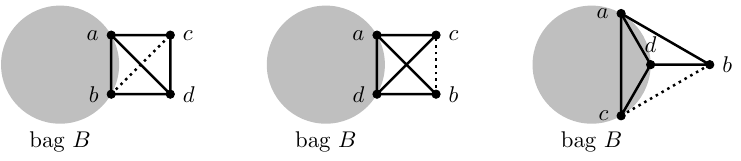}
\end{center}
\caption{Illustration of Lemma~\ref{lem:diamond_inside_bag}. Each case violates Lemma~\ref{lem:bag-decomposition}. For instance, in the first case, if the edges~$\edge{a}{d}$ and~$\edge{b}{d}$ are in the same bag, Condition~\ref{p:edg} is violated; otherwise, Condition~\ref{p:nonadj} is violated.}
\label{fig-diamond-in-a-bag}
\end{figure}

\begin{lemma}
\label{lem_claw_inside_bag}
Let~$H$ be a subgraph (not  necessarily  induced) of~$G$ isomorphic to a claw. % such that the center of~$H$ is in $G-X$.
Let $B\in\cB(G-X)$ be an unattached bag containing at least two leaves of~$H$.
Then~$B$ contains the center of~$H$.
\end{lemma}

\begin{proof}
Let~$c$ be the center of~$H$, and let~$\ell_1$ and~$\ell_2$ be two of the leaves of~$H$ included in~$B$. As~$B$ is a bag, there is an edge between~$\ell_1$ and~$\ell_2$ in~$G$.
For the sake of contradiction, assume that $c\not\in B$. We distinguish between two cases.
\begin{description}
\item[Case~1:] $c\in X$. \hfill

As~$B$ is unattached, there must be another vertex $u\not\in \{\ell_1, \ell_2\}$ in the bag~$B$ which is not adjacent to~$c$. However,~$c$,~$\ell_1$,~$\ell_2$, and~$u$ form an induced diamond which is edge-disjoint with all claws and diamonds in~$G[X]$, contradicting with the definition of~$X$.

\item[Case~2:] $c\in V(G)\setminus X$.

Let~$B'$ be the bag including the edge~$\edge{c}{\ell_1}$. Due to Condition~\ref{p:nonadj} of Lemma~\ref{lem:bag-decomposition},~$\edge{c}{\ell_2}$ must be also in~$B'$, which further implies that the edge $\edge{\ell_1}{\ell_2}$ is also contained in~$B'$. As we assumed that~$c$ is not in~$B$, we know that~$B$ and~$B'$ are distinct. This contradicts with Condition~\ref{p:edg} of
Lemma~\ref{lem:bag-decomposition} because~$B$ and~$B'$ share at least two common vertices.
\end{description}
As both cases lead to some contradiction, we know that $c\in B$.
\end{proof}

\begin{lemma}[\cite{DBLP:journals/mst/CyganPPLW17}]
\label{lem_untached_bag_cliques_and_isolated_vertices}
Let~$B$ be an unattached bag in $\cB(G-X)$ and let~$S$ be a minimal CDH set of~$G$. Then, $G[B]-S$ consists of a clique and~$i$ isolated vertices for some nonnegative integer~$i$.
\end{lemma}

Let~$B'$ be the set of nonisolated vertices in $G[B]-S$, where~$B$ is as stipulated in Lemma~\ref{lem_untached_bag_cliques_and_isolated_vertices}. In fact, if~$B'$ is not a clique in $G[B]-S$, one can show that $S\setminus E(B)$ is a smaller CDH set of~$G$. We refer the formal proof of Lemma~\ref{lem_untached_bag_cliques_and_isolated_vertices} to~\cite[Lemma~3.6]{DBLP:journals/mst/CyganPPLW17}.

The next property says that for every bag $B\in \cB(G-X)$ and every vertex $x\in X$, $|N(x)\cap B|\in \{0, 1, |B|\}$ holds. In fact, if~$v$ is adjacent to more than one vertex of~$B$ but not all of them, then there is an induced diamond (formed by~$v$, two of~$v$'s neighbors in~$B$, and one vertex in~$B$ which is not adjacent to~$v$)  which is edge-disjoint from all induced claws and diamonds in~$G[X]$, a contradiction. The following lemma summarizes this property.

\begin{lemma}[\cite{DBLP:journals/mst/CyganPPLW17}]
\label{lem-every-vertex-in-x-cannot-adjacenet-to-only-two-vertices-in-a-bag}
If a vertex $x\in X$ is adjacent to two vertices in a bag $B\in \cB(G-X)$, then~$B$ is attached to~$x$.
\end{lemma}

Next, we study a property regarding outlier bags.

\begin{lemma}
\label{lem-outlier-bags-donot-connect-to-modulator}
Let $B\in \cB(G-X)$ be an outlier bag. Then, none of the vertices in~$B$ is adjacent to any vertex in~$X$.
\end{lemma}

\begin{proof}
Assume, for the sake of contradiction, that~$B$ contains a vertex~$v$ who has a neighbor~$x\in X$.
If~$v$ is a simplicial vertex in $G-X$, then either~$\{v\}$ is attached to~$v$, or~$B$ is attached to~$v$ (Lemma~\ref{lem-every-vertex-in-x-cannot-adjacenet-to-only-two-vertices-in-a-bag}),  contradicting that~$B$ is an outlier bag. If~$v$ is not a simiplicial vertex, then $\abs{B}>1$ and~$v$ is contained in another bag $B'\in \cB(G-X)$ such that $\abs{B'}>1$. Let $u$ be a vertex in $B\setminus \{v\}$, and let~$w$ be a vertex in $B'\setminus \{v\}$. Due to Condition~\ref{p:nonadj} of Lemma~\ref{lem:bag-decomposition},~$u$ is not adjacent to~$w$ in~$G$. Due to Lemma~\ref{lem-every-vertex-in-x-cannot-adjacenet-to-only-two-vertices-in-a-bag},~$u$ is not adjacent to~$x$. By the same lemma,~$w$ is not adjacent to~$x$, since otherwise~$B'$ is attached and hence~$B$ cannot be an outlier bag. However, this means that~$v$,~$u$,~$w$, and~$x$ form an induced claw in~$G$. As only one vertex (i.e., $x$) of the induced claw is in~$X$, this contradicts that~$X$ is a modulator.
\end{proof}

We are ready to prove the soundness of Rules~\ref{rule_remove_edges_in_outlier_bags} and \ref{rule_deleting_edges_in_irrevalent_bags}. These two rules share some common principle and hence can be proved in a similar manner.
%In particular, let~$B$ be a bag such that edges in~$E(B)$ are deleted in Rule~\ref{rule_remove_edges_in_outlier_bags} (resp.\ Rule~\ref{rule_deleting_edges_in_irrevalent_bags}). Then, for any minimal CDH set~$S$ of~$G$ of size at most~$k$, the vertices in each bag sharing a common vertex with~$B$ induce a graph consisting of a clique and several isolated vertices in~$G-S$.

\begin{lemma}
\label{deletevertex}
Rules~\ref{rule_remove_edges_in_outlier_bags} and \ref{rule_deleting_edges_in_irrevalent_bags} are sound.
\end{lemma}

\begin{proof}
Let~$B$ be a bag as stipulated in Rule~\ref{rule_remove_edges_in_outlier_bags} (resp.\ Rule~\ref{rule_deleting_edges_in_irrevalent_bags}).
We show that~$(G, k)$ is a {\yesins} if and only if~$(G', k)$ is a {\yesins}. Let~$F$ be the set of all bags in $G-X$ of size at least two that share some vertex with~$B$. We distinguish between two cases: $F=\emptyset$ and $F\neq \emptyset$.

First, if $F=\emptyset$, then~$B$ is a clique and a connected component of~$G$, and it is clear that a clique does not contain any induced claws and diamonds. So, we can safely remove~$B$ from the graph~$G$. In the reduction rules, after removing edges in~$E(B)$, vertices in~$B$ become isolated vertices and are further removed by Rule~\ref{rule_remove_components}.

Now we consider the second case where $F\neq\emptyset$.
Let~$S$ be a minimal CDH set of~$G$ of size at most~$k$. Due to the definition of~$B$, each bag in~$F$ is unattached (resp.\ either a big attached bag or an unattached bag). Hence, Due to Lemmas~\ref{lem:large_bag_disjiont_with_solution} and~\ref{lem_untached_bag_cliques_and_isolated_vertices}, every bag in $\{B\}\cup F$ induces a graph consisting of a clique and (possibly) some isolated vertices in $G-S$. We show that $S\setminus E(B)$ is a CDH set of~$G'$. To this end, we need only to show that $G-S-E(B)$ is \{claw, diamond\}-free. For the sake of contradiction, assume that there is an induced claw or diamond~{\fst} in $G-S-E(B)$. Observe that~$B$ contains at least two vertices of~{\fst}, since otherwise~{\fst} exists in $G-S$, contradicting that~$S$ is a CDH set of~$G$. Moreover, the vertices of~{\fst} in~$B$ must form an independent set of~{\fst}. Consider first the case where~{\fst} is an induced claw. Due to the above discussion, the center of {\fst} cannot be in~$B$.
%Moreover, due to Lemma~\ref{lem-every-vertex-in-x-cannot-adjacenet-to-only-two-vertices-in-a-bag}, the center cannot be in~$X$. This leaves only the possibility that the center of~{\fst} is in some bag $B'\in F$. Moreover,~$B$ includes at least two leaves of~{\fst}.
However, this contradicts with Lemma~\ref{lem_claw_inside_bag}.
Hence,~{\fst} cannot be an induced claw. Now, we consider the case where~{\fst} is an induced diamond. As discussed above,~$B$ contains at least two vertices of~{\fst}. Then, due to Lemma~\ref{lem:diamond_inside_bag}, all vertices of~{\fst} are in~$B$. However, this cannot be the case as~{\fst} contains edges but the vertices in~$B$ induce an independent set in $G-S-E(B)$. This completes the proof for this direction.

It remains to prove the other direction.
Let~$S'$ be a minimal CDH set of~$G'$ of size at most~$k$. We show that~$S'$ is also a CDH set of~$G$. To this end, we need only to show that~$G'-S'+E(B)=G-S'$ is $\{\text{claw, diamond}\}$-free. For the sake of contradiction, assume that there is an induced claw or diamond~{\fst} in $G-S'$. Clearly,~$B$ contains at least two vertices of~{\fst}, since otherwise~{\fst} exists in $G'-S'$, a contradiction.

Consider first the case where~{\fst} is an induced claw. Obviously,~$B$ can contain at most one leaf of~{\fst}. Hence,~$B$ contains the center~$c$ and a leaf~$\ell$ of~{\fst}. Let~$\ell_1$ and~$\ell_2$ be the other two leaves of {\fst}. We claim that no matter whether~$B$ is an outlier bag (in Rule~\ref{rule_remove_edges_in_outlier_bags}) or a border bag (in Rule~\ref{rule_deleting_edges_in_irrevalent_bags}),~$\ell_1$ and~$\ell_2$ are both in $G-X$. This is true for the former case due to Lemma~\ref{lem-outlier-bags-donot-connect-to-modulator}.
Now we consider the latter case. For the sake of contradiction, assume that~$\ell_i$ for some $i\in \{1,2\}$ is in~$X$. Let~$B'$ be the other bag including~$c$. If $|B'|=1$, then~$B'$ is attached to~$\ell_i$ in~$G$, contradicting that~$B$ is a border bag without common vertices with small attached bags, as stipulated in Rule~\ref{rule_deleting_edges_in_irrevalent_bags}. Hence, let us assume that $|B'|>1$. Then,~$B'$ must be attached to~$\ell_i$ in~$G$, since otherwise~$\ell$,~$\ell_i$,~$c$, and any vertex in~$B'$ which is not adjacent to~$\ell_i$ is an induced claw in~$G$ which is edge-disjoint with any induced claws and diamonds in~$G[X]$, a contradiction. This means that~$B'$ is a big bag. We continue the proof of the claim by considering the location of~$\ell_{3-i}$. If $\ell_{3-i}\in X$, then by replacing occurrences of~$i$ with~$3-i$ in the above argument, we can conclude that~$B'$ is adjacent to~$\ell_{3-i}$ in~$G$ too. Then, $B'\cup \{\ell_1,\ell_2\}$ must be a clique in~$G$, since otherwise there is an induced diamond (formed by~$\ell_1$,~$\ell_2$, and any two vertices in~$B'$) which is edge-disjoint with all induced claws and diamonds in~$G[X]$, a contradiction. Due to the definition of~$G'$,~$B'$ is also a clique in~$G'$. Then, according to Lemma~\ref{lem:large_bag_disjiont_with_solution},~$S'$ is disjoint from all edges in $E(B'\cup \{\ell_1,\ell_2\})$, which contradicts that~$\ell_1$ and~$\ell_2$ are two leaves of~{\fst} in $G-S'$. Assume now that~$\ell_{3-i}$ is in $G-X$. Then, due to Conditions~\ref{p:ver} and~\ref{p:edg} of Lemma~\ref{lem:bag-decomposition}, the edge $\edge{c}{\ell_{3-i}}$ must be included in the bag~$B'$, implying that $B'\cup \{\ell_1,\ell_2\}$ is a clique in~$G$ (and~$G'$). However, this contradicts that~$S'$ is disjoint from $E(B'\cup \{\ell_1,\ell_2\})$. This completes the proof for the claim that both~$\ell_1$ and~$\ell_2$ are in $G-X$. Then, due to this claim and  Conditions~\ref{p:ver} and~\ref{p:edg} of  Lemma~\ref{lem:bag-decomposition},~$\ell_1$ and~$\ell_2$ are both in~$B'$.
Notice that~$X$ is also a modulator of~$G'$.  %and hence the two lemmas apply to~$G'$.
Then, if~$B'$ is unattached, due to Lemma~\ref{lem_untached_bag_cliques_and_isolated_vertices},~$G'[B']-S'$ consists of a clique and (possibly) several isolated vertices. However,~$c$ is adjacent to both~$\ell_1$ and~$\ell_2$, but~$\ell_1$ and~$\ell_2$ are not adjacent in~$G-S'$, a contradiction. % (recall that $G[B']=G'[B']$).
If~$B'$ is attached, then~$B$ cannot be an outlier bag. So, in this case, we analysis only for Rule~\ref{rule_deleting_edges_in_irrevalent_bags}. As stipulated in this rule, we know that~$B'$ is a big attached bag. Then, from Lemma~\ref{lem:large_bag_disjiont_with_solution},~$S'$ does not contain any edge in~$G[B']$, and hence~$B'$ is still a clique in $G-S'$. However, this contradicts that~$\ell_1$ and~$\ell_2$ are two leaves in an induced claw in $G-S'$.

% However, this contradicts with Lemmas~\ref{lem:large_bag_disjiont_with_solution} and~\ref{lem_untached_bag_cliques_and_isolated_vertices}.
%, which imply that~$D$ induces a subgraph consisting of a clique and $i\geq 0$ isolated vertices in $G'-S'$ but $c, \ell_1, \ell_2$ induce a path

Consider now that {\fst} is an induced diamond in $G-S'$. Due to Lemma~\ref{lem:diamond_inside_bag}, all vertices of~{\fst} are in~$B$. However,~$B$ induces a clique in $G-S'$, a contradiction too.
\end{proof}

The next reduction rule reduces the size of attached bags.

\begin{myrule}
\label{rule_remove_one_vertex_from_a_big_bag}
If there is an attached bag~$B$ which is of size at least~$\nbg$ and shares a vertex~$v$ with a border bag~$B'$, then delete all edges incident to~$v$ in $E(B\cup A(B))$.
\end{myrule}

We refer to Figure~\ref{fig-lem-12-aa} for an illustration of Rule~\ref{rule_remove_one_vertex_from_a_big_bag}.
\begin{figure}
\centering
{
\includegraphics[width=0.85\textwidth]{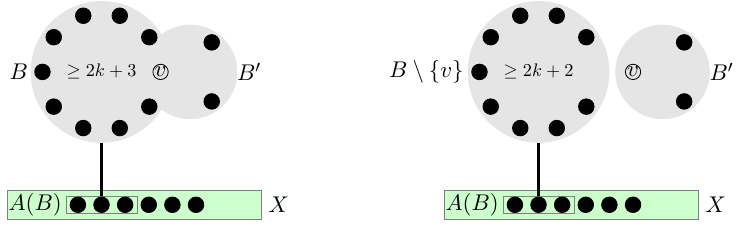}
}
\caption{An illustration of  Rule~\ref{rule_remove_one_vertex_from_a_big_bag}.
%the proof of Lemma~\ref{lem_rule_remove_one_vertex_from_a_big_bag_is_sound}. The two instances on the left side and on the right side are respectively the instances before and after the application of Rule~\ref{rule_remove_one_vertex_from_a_big_bag}.
}
\label{fig-lem-12-aa}
\end{figure}
A special case of Rule~\ref{rule_remove_one_vertex_from_a_big_bag} is that when~$v$ is a simplicial vertex in $G-X$. In this case, after the application of Rule~\ref{rule_remove_one_vertex_from_a_big_bag},~$v$ becomes an isolated vertex. Then, an application of Rule~\ref{rule_remove_components} deletes~$v$ from~$G$.

Now we prove the soundness of Rule \ref{rule_remove_one_vertex_from_a_big_bag}. We claim that in Rule~\ref{rule_remove_one_vertex_from_a_big_bag} it holds that $N(v)\cap X=A(B)$. For the sake of contradiction, assume that there is an $x\in X\setminus A(B)$ such that $\edge{v}{x}\in E(G)$. If~$v$ forms a bag itself, then due to Condition~(1) of Lemma~\ref{lem:bag-decomposition},~$B$ and~$B'$ are the only two bags including~$v$ and, moreover,  $B'=\{v\}$. This implies that~$B'$ is attached to~$x$, contradicting that~$B'$ is a border bag. If, however,~$B'\setminus \{v\}\neq \emptyset$, then~$v$,~$x$, any vertex from~$B$, and any vertex from~$B'\setminus \{v\}$ induce a claw, which is edge-disjoint from all induced claws and diamonds in~$G[X]$, a contradiction.

\begin{lemma}
\label{lem_rule_remove_one_vertex_from_a_big_bag_is_sound}
Rule~\ref{rule_remove_one_vertex_from_a_big_bag} is sound.
\end{lemma}

\begin{proof}
%Let $(G', k)$ be the instance after applying Rule~\ref{rule_remove_one_vertex_from_a_big_bag} to $(G, k)$.
Let~$B$,~$B'$, and~$v$ be as stipulated in Rule~\ref{rule_remove_one_vertex_from_a_big_bag}.
In the following, we show that any minimal CDH set of~$G$ of size at most~$k$ is a CDH set of~$G'$, and vice versa.

$(\Rightarrow)$ Let~$S$ be a minimal CDH set of~$G$ of size at most~$k$. We show that~$S$ is a CDH set of~$G'$. Due to Lemma~\ref{lem:large_bag_disjiont_with_solution},~$S$ and $E(B\cup A(B))$ are disjoint. Moreover, due to Lemma~\ref{lem_untached_bag_cliques_and_isolated_vertices}, $(G-S)[B']$ consists of a clique and (possibly) several isolated vertices. We show that no induced claw or diamond occur after deleting all edges incident to~$v$ in $E(B\cup A(B))$ from $G-S$, i.e., in~$G'-S$. For the sake of contradiction, assume that there is an induced claw or diamond~{\fst} in $G'-S$. We distinguish between the following cases.
\begin{figure}
\centering
{
\includegraphics[width=0.75\textwidth]{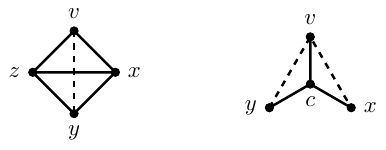}
}
\caption{Illustration of Case~1 (left) and Case~2 (right) in the proof of Theorem~\ref{lem_rule_remove_one_vertex_from_a_big_bag_is_sound}.}
\label{fig-lem-12}
\end{figure}
\begin{description}
\item {\bf{Case~1:}} {\fst} is a diamond.

Without loss of generality, let the vertices of {\fst} be~$v$,~$x$,~$y$, and~$z$ with the edge between~$v$ and~$y$ missing (see Figure~\ref{fig-lem-12}). So, we know that $y\in B\cup A(B)$.

If $y\in A(B)$, as~$B'$ is unattached, there is a vertex $u\in B'$ which is not adjacent to~$y$ in~$G$. We claim that neither of~$x$  and~$z$ is in the modulator~$X$. Assume for the sake of contradiction that $x\in X$. Then, as $N(v)\cap X=A(B)$, and~$v$ and~$x$ are adjacent in~$G$, it holds that $x\in A(B)$. However, in this case the edge between~$v$ and~$x$ cannot be in {\fst} (it is removed by Rule~\ref{rule_remove_one_vertex_from_a_big_bag} and hence not in~$G'$), a contradiction. Due to symmetry, we can show that~$z\not\in X$.
%as~$B'$ is unattached,~$u$ is not adjacent to~$x$ in~$G$. If~$B$ is not attached to~$x$, then~$v$,~$u$,~$x$, and any vertex in $B\setminus v$ form an induced claw in~$G$, contradicting that~$X$ is a modulator. The same analysis applies to $z$ due to symmetry.
Then, as both~$x$ and~$z$ are adjacent to~$v$, from Conditions~\ref{p:ver} and~\ref{p:edg} of Lemma~\ref{lem:bag-decomposition}, it holds that $x, z\in B\cup B'$. However, neither of~$x$ and~$z$ can be in~$B$, since otherwise the edge between~$v$ and~$x$ (if $x\in B$), or the one between~$v$ and~$z$ (if $z\in B$) is removed by Rule~\ref{rule_remove_one_vertex_from_a_big_bag} and cannot be in {\fst}. Neither of them can be in~$B'$ either, since otherwise~$B'$ is attached to~$y$, a contradiction too.

Now we consider the case where $y\in B$. First, neither of~$x$ and~$z$ can be in~$B$, since otherwise they have been removed by Rule~\ref{rule_remove_one_vertex_from_a_big_bag}. They cannot be in~$B'$ according to Condition~\ref{p:nonadj} of Lemma~\ref{lem:bag-decomposition}. It remains only the case that both~$x$ and~$z$ are in the modulator~$X$. Then, as $N(v)\cap X=A(B)$, we know that $x,z\in A(B)$. However, in this case the edges between~$v$ and $\{x, z\}$ cannot be in {\fst} since they have been removed by Rule~\ref{rule_remove_one_vertex_from_a_big_bag}.

\item {\bf{Case~2:}} {\fst} is a claw.

Without loss of generality, let the vertices of {\fst} be~$v$,~$x$,~$y$, and~$c$ (see Figure~\ref{fig-lem-12}). Since {\fst} occurs only after deleting some edges as stipulated in Rule~\ref{rule_remove_one_vertex_from_a_big_bag}, it must be that~$v$ is a leave of {\fst}. Without loss of generality, let us assume that~$c$ is the center of {\fst}. So, at least one of~$x$ and~$y$ is in $B\cup A(B)$. By symmetry, let us assume that $x\in B\cup A(B)$. Similar to the above analysis, we further consider two cases and show that both cases lead to contradictions.

We consider first the case where $x\in A(B)$. Note that~$c$ cannot be in $B\cup A(B)$, since otherwise the edge between~$v$ and~$c$ has been removed by Rule~\ref{rule_remove_one_vertex_from_a_big_bag}. It cannot be in~$B'$ either, since otherwise~$B'$ is attached to~$x$ by Lemma~\ref{lem-every-vertex-in-x-cannot-adjacenet-to-only-two-vertices-in-a-bag}. Therefore, it must be that $c\in X\setminus A(B)$. However, as~$v$ is adjacent to~$c$, this contradicts with $N(v)\cap X=A(B)$.
%Then, there exists at least one vertex in~$B$ who is not adjacent to~$c$ in~$G$. Moreover, as~$B'$ is not an attached bag, there is at least one vertex in~$B'$ who is not adjacent to~$c$ in~$G$ too. However, these two vertices together with~$v$ and~$c$ form an induced claw in~$G$ that is edge-disjoint from all claws and diamonds contained in~$G[X]$, contradicting that~$X$ is a modulator.

Now we consider the second case where $x\in B$. Similar to the above analysis, we can first claim that $c\not\in B\cup A(B)$. In addition, by Condition~\ref{p:nonadj} of Lemma~\ref{lem:bag-decomposition}, $c\not\in B'$. So, it must be that $c\in X$. However, as both~$v$ and~$x$ are adjacent to~$c$ in {\fst}, we know that $c\in A(B)$ by Lemma~\ref{lem-every-vertex-in-x-cannot-adjacenet-to-only-two-vertices-in-a-bag}, a contradiction.
 %However, in this case, the edge between~$c$ and~$v$ has been removed by Rule~\ref{rule_remove_one_vertex_from_a_big_bag}.
\end{description}

$(\Leftarrow)$ Now we prove the other direction. Let~$S'$ be a minimal CDH set of~$G'$ of size at most~$k$. Notice that~$X$ is a modulator of~$G'$% \textcolor{blue}{please check this!!!}, and the bags in $G'-X$ are the ones in $G$ with~$B$ being replaced with~$B\setminus \{v\}$%
.
Moreover, $B\setminus \{v\}\cup A(B)$ is a big bag in~$G'$. Hence, due to Lemma~\ref{lem:large_bag_disjiont_with_solution},~$S'$ and $E((B\setminus \{v\})\cup A(B))$ are disjoint, meaning that $(B\setminus \{v\})\cup A(B)$ remains as a clique in $G'-S'$. Moreover, $(G'-S')[B']$ consists of a clique and (possibly) several isolated vertices.
We show now that adding all edges incident to~$v$ in $E(B\cup A(B))$ to $G'-S'$ does not result in induced claws or diamonds. Assume, for the sake of contradiction, that after adding these edges there is an induced claw or diamond~{\fst}. By symmetry, we have~11 possibilities over~{\fst} to consider, as depicted in Figure~\ref{fig-lem-12-c}.
%Our proof proceeds as follows.

\begin{figure}
\centering
{
\includegraphics[width=0.85\textwidth]{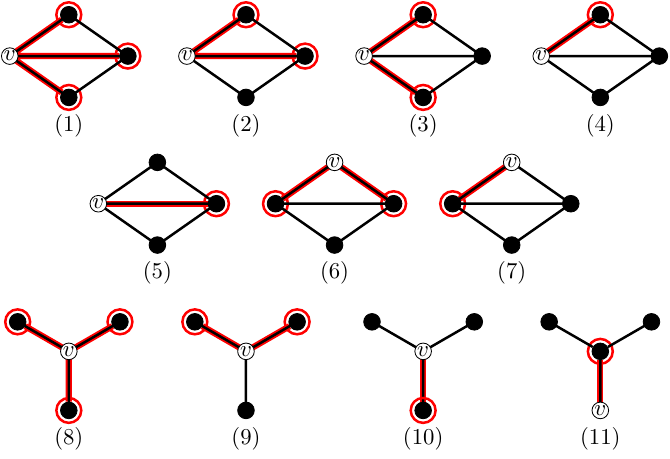}
}
\caption{All possible cases of {\fst} in the proof of Theorem~\ref{lem_rule_remove_one_vertex_from_a_big_bag_is_sound} for the~$(\Leftarrow)$ direction. Vertices circulated by a red circle belong to $B\cup A(B)$ (i.e., vertices in $F_1$).}
\label{fig-lem-12-c}
\end{figure}

Let~$F_1=(V$({\fst})$\setminus \{v\})\cap (B\cup A(B))$ be the subset of vertices in {\fst}  except~$v$ that are from $B\cup A(B)$. Moreover, let $F_2=V$({\fst})$\setminus (F_1\cup \{v\})$. Clearly, $F_1\neq \emptyset$ (otherwise~{\fst} exists in $G'-S'$, contradicting that~$S'$ is a CHD set of~$G'$) and $F_1$ is a clique (since $F_1\subseteq (B\setminus \{v\}\cup A(B))$ which is a clique as discussed above) in $G'-S'$. We show that all the~11 cases shown in~Figure~\ref{fig-lem-12-c} are impossible. %Let $u\in B\setminus \{v\}$ be any vertex in $B\setminus \{v\}$.
Keep in mind that $(B\setminus \{v\})\cup A(B)$ is a big clique, and hence every vertex in $B\setminus \{v\}$ is adjacent to every vertex in~$F_1$ in $G'-S'$.
\begin{itemize}
\item Cases (1), (3), (8), (9) are impossible because in these cases~$F_1$ is not a clique.
%\item As $(B\setminus \{v\})\cup A(B)$ is a clique in $G'-S'$ and $N_G(v)\cap X=A(B)$, all neighbors of~$v$ in $F_1$ are also neighbors o~$u$ in $F_2$. Therefore, for Cases~(1) and~(8), replacing~$v$ with~$u$ in {\fst} obtains another claw or diamond in $G'-S'$, contradicting that~$S'$ is a CDH set of~$G'$.

\item Now we consider Cases~(2) and~(6). For both cases, we have that $\abs{F_2}=1$. Let~$F_2=\{w\}$ and~$F_1=\{x, y\}$ in each case.

In Case~(2), as $w\not\in F_1$ and $N(v)\cap X=A(B)$, we know that $w\not\in X$, meaning that~$w$ is from some bag of $\cB(G-X)$ that contains~$v$. By Conditions~\ref{p:ver} and~\ref{p:edg} of Lemma~\ref{lem:bag-decomposition} and the fact that $w\not\in F_1\subseteq B$, it holds that $w\in B'$. As~$w$ is adjacent to one of~$x$ and~$y$, say~$x$, by Condition~\ref{p:nonadj} of Lemma~\ref{lem:bag-decomposition},~$x\not\in B$. It follows that $x\in A(B)\subseteq X$. However, in this case $B'$ is attached to $x$, a contradiction.

 In Case~(6),~$v$ and~$w$ are not adjacent. If $w\in X$, there is at least one vertex in~$B\setminus \{v\}$, say,~$u$, which is not adjacent to~$w$, since otherwise by Lemma~\ref{lem-every-vertex-in-x-cannot-adjacenet-to-only-two-vertices-in-a-bag},~$B$ is attached to~$w$ which contradicts that~$v$ and~$w$ are not adjacent. Then,  we can obtain an induced  diamond in~$G'-S'$ by replacing~$v$ with~$u$ in~{\fst}, contradicting that~$S'$ is a CDH set of~$G'$. If~$w\not\in X$, then~$w$ is from some bag $C\in \cB(G-X)$. Let~$u'$ be any arbitrary vertex in $B\setminus \{v\}$.  As~$v$ is not adjacent to~$w$, it holds that $C\neq B$. Then, by Condition~\ref{p:nonadj} of Lemma~\ref{lem:bag-decomposition},~$u'$ is not adjacent to~$w$. In this case, replacing~$v$ with~$u'$ in {\fst} gives us an induced diamond in $G'-S'$, a contradiction too.

\item Now we consider Cases~(4), (5), (7), (10), and (11). Let $F_1=\{w\}$, and $F_2=\{x,y\}$ in each case.

For Cases (4), (5) and~(7), it holds that $x, y\not\in B\cup A(B)$. We know then that $x,y\in B'$. By Condition~\ref{p:nonadj} of Lemma~\ref{lem:bag-decomposition},~$w$ cannot be in~$B$. By the definition of~$F_1$, it must be that $w\in A(B)$. However, by Lemma~\ref{lem-every-vertex-in-x-cannot-adjacenet-to-only-two-vertices-in-a-bag},~$B'$ is attached to~$w$, a contradiction.

For Case~(10), as both~$x$ and~$y$ are adjacent to~$v$, and they are not from~$F_1$, it holds that $x, y\in B'$. However, as~$x$ is not adjacent to~$y$, this contradicts that $(G'-S')[B']$ consists of a clique and (possibly) several isolated vertices.

For Case~(11), if $w\in B$, %then by Point~\ref{p:nonadj} of Lemma~\ref{lem:bag-decomposition}, $x, y\not\in B'$. Then,  at least one of~$x$ and~$y$, say~$x$, must be from~$X$, since otherwise~$w$ would be in three bags in $\cB(G-X$), contradicting with Point~\ref{p:ver} of Lemma~\ref{lem:bag-decomposition}. L
let~$u$ be any arbitrary vertex from $B\setminus \{v, w\}$. As $B\setminus \{v\}$ is a big clique in~$G'$, $S'$ does not contain any edge in $E(B\setminus \{v\})$. So,~$u$ is adjacent to~$w$ in $G'-S'$. Moreover,~$w$ is not adjacent to any of~$x$ and~$y$ in $G'-S'$. Suppose for the sake of contradiction that~$w$ is adjacent to~$x$ (resp.\ $y$). If $x\in X$ (resp.\ $y\in X$),~$B$ is attached to~$x$ (resp.\ $y$), and hence it holds that $x\in A(B)$ (resp.\ $y\in A(B)$). This contradicts that $x\in F_2$ (resp.\ $y\in F_2$). If $x\not\in X$ (resp.\ $y\not\in X$), then~$x$ is in some bag other than~$B$ and~$B'$. However, this contradicts with Condition~\ref{p:nonadj} of Lemma~\ref{lem:bag-decomposition}. Now, it is easy to see that after replacing~$v$ with~$u$ in {\fst}, we obtain another induced claw formed by~$u$,~$w$,~$x$, and~$y$ in  $G'-S'$, contradicting that~$S'$ is a CDH set of~$G'$.
If $w\in X$, we first show that at most one vertex from~$B$ is adjacent to~$x$, and at most one vertex from~$B$ is adjacent to~$y$. By symmetry, we only give the proof for~$x$. If~$x\in X$, then  as $x\not\in F_1$, we know that~$B$ is not attached to~$x$. Then, by Lemma~\ref{lem-every-vertex-in-x-cannot-adjacenet-to-only-two-vertices-in-a-bag}, at most one vertex in~$B$ is adjacent to~$x$. Otherwise,~$x$ is from some bag, and moreover, this bag is neither~$B$ (since $x\in F_2$) nor~$B'$ (otherwise~$B'$ is attached to~$w$). So, by Lemma~\ref{lem:bag-decomposition}, none of the vertices in~$B$ is adjacent to~$x$. Finally, as $B\setminus \{v\}$ is a big bag, we know that there is at least one vertex $u\in B$ which is not adjacent to any of~$x$ and~$y$ in $G'-S'$. Then, replacing~$v$ with~$u$ in {\fst} offers us a new induced claw in $G'-S'$, a contradiction.
\end{itemize} %We refer to Figure~\ref{fig-lem-12} for an illustration of the proof.

This completes the proof that Rule~\ref{rule_remove_one_vertex_from_a_big_bag} is sound.
\end{proof}

Finally, we study a reduction rule to bound the size of each border bag.

\begin{myrule}
\label{Rule_replacing_D_2}
If there is a border bag~$B$ of size at least $2k+3$, then delete all edges in $E(B)$ and, moreover, for each attached bag that shares a vertex~$v$ with~$B$, add $2k+1$ new vertices and add edges so that these newly added vertices and~$v$ form a clique.
\end{myrule}

\begin{lemma}[\cite{DBLP:journals/mst/CyganPPLW17}]
\label{lem-only-one-bag-including-a-vertex-is-attached}
Let~$v$ be a vertex in $G-X$ adjacent to a vertex $x\in X$.
Then there is exactly one bag in $\cB(G-X)$ that contains~$v$ and is attached to~$x$.
\end{lemma}

The main idea of the proof of Lemma~\ref{lem-only-one-bag-including-a-vertex-is-attached} is as follows. If both bags including~$v$, say~$B$ and~$B'$,  were attached to~$x$ (observe that $B\setminus \{v\}\neq\emptyset$ and $B'\setminus \{v\}\neq \emptyset$ hold), then, one can check that~$x$,~$v$, any vertex from $B\setminus \{v\}$, and any vertex from $B'\setminus \{v\}$ induce a diamond that is edge-disjoint from all induced claws and diamonds in~$G[X]$, a contradiction.

Armed with Lemma~\ref{lem-only-one-bag-including-a-vertex-is-attached}, we are ready to prove the soundness of Rule~\ref{Rule_replacing_D_2}.

\begin{lemma}
\label{lem_rule_replacing_D_2}
Rule~\ref{Rule_replacing_D_2} is sound.
\end{lemma}

\begin{proof}
Let~$B$ be a bag as stipulated in Rule \ref{Rule_replacing_D_2}. For each attached bag~$B'$ sharing a vertex with~$B$, let %$v(B')$ be the common vertex of~$B$ and~$B'$, and
~$C(B')$ be the set of the~$2k+1$ newly introduced vertices for~$B'$. Let~$C$ be the set of all newly introduced vertices in Rule~\ref{Rule_replacing_D_2}. We prove the soundness as follows.

Let~$S$ be a minimal CDH set of~$G$ of size at most~$k$. We claim that~$S$ is a CDH set of~$G'$. For the sake of contradiction, assume that this is not the case, and let {\fst} be an induced claw or diamond in $G'-S$. Clearly, at least two vertices of {\fst} are in $B\cup C$, since otherwise {\fst} exists in $G-S$, contradicting that~$S$ is a CDH set of~$G$. 
We claim that at most one vertex of~{\fst} can be in~$B$.  In fact, as~$B$ is unattached in~$G$,  due to Lemmas~\ref{lem:diamond_inside_bag} and~\ref{lem_claw_inside_bag}, if~$B$ contains at least two vertices of {\fst}, then all the vertices of {\fst} are in~$B$ if {\fst} is a diamond, and the center and at least one leaf of {\fst} are in~$B$ if~{\fst} is a claw, which contradicts that~$B$ is an independent set of~$G'-S$. So, the claim holds.  
Let~$K$ be the set of attached bags~$B'$ sharing a vertex with~$B$ such that~$C(B')$ contains at least one vertex of~{\fst}. The above discussions imply that $K\neq \emptyset$. We claim that~$|K|=1$. For the sake of contradiction, assume that $|K|\geq 2$. Let~$B_1$ and~$B_2$ be any two bags in~$K$. Clearly, the distance between any vertex in~$C(B_1)$ and any vertex in~$C(B_2)$ is at least~$3$ in~$G'$. However, the distance between every two vertices in~{\fst} is at most~$2$, a contradiction. So, let~$B'$ be the only bag in~$K$ and~$v$ the common vertex of~$B$ and~$B'$. Let~$D$ be the set of vertices of~{\fst} in~$C(B')$. Due to the definition of~$C(B')$,~$v$ is the only vertex in~$G$ which is adjacent to vertices in~$C(B')$. In other words,~$v$ separates~$C(B')$ from all the other vertices. This implies that~{\fst} is a claw and~$v$ is the center of~{\fst}. As~$C(B')$ form a clique in~$G'-S$, we know that~$D$ is a singleton consisting of one of the leaves of~{\fst}. Without loss of generality, let $D=\{\ell\}$, and let~$\ell_1$ and~$\ell_2$ denote the other two leaves of~{\fst}. Due to Conditions~\ref{p:ver} and~\ref{p:edg}  of Lemma~\ref{lem:bag-decomposition}, it holds that $\ell_1, \ell_2\in B\cup B'$. As~$v\in B$ and~$B$ is an independent set in~$G'-S$, it follows that $\ell_1, \ell_2\in B'$ (and we know that the edge between~$\ell_1$ and~$\ell_2$ is contained in~$S$). Let~$u$ be any arbitrary vertex in $B\setminus \{v\}$. Due to Condition~\ref{p:nonadj} of Lemma~\ref{lem:bag-decomposition},~$u$ is not adjacent to any of~$\ell_1$ and~$\ell_2$ in~$G$. This implies that~$\ell_1$,~$\ell_2$,~$v$, and~$u$ form an induced claw in $G-S$, contradicting that~$S$ is a CDH set of~$G$. 

We now prove the opposite direction. Let~$S$ be a minimal CDH set of~$G'$ of size at most~$k$. Obviously, $G'-S-C$ is still \{claw, diamond\}-free. Due to Lemma~\ref{lem:large_bag_disjiont_with_solution},~$S$ excludes all edges between vertices in~$C$. We claim that~$S$ is a CDH set of~$G$. Assume that this is not the case, and let~{\fst} be a forbidden structure in $G-S$. Hence,~$B$ includes at least two vertices of~{\fst}, since otherwise~{\fst} exists in $G'-S-C$, a contradiction. Then, if~{\fst} is an induced diamond, due to Lemma~\ref{lem:diamond_inside_bag}, all vertices of~{\fst} are in~$B$, contradicting that~$B$ is a clique in $G-S$. If~{\fst} is an induced claw, then it must be that the center of~{\fst} and exactly one leaf of~{\fst} are in~$B$ (as~$B$ is a clique in $G-S$). Let~$c$ be the center and~$\ell$ be the leaf. Let~$B'$ be the other bag including~$c$. Then, replacing~$\ell$ with any vertex in~$C(B')$ in~{\fst} leads to an induced claw in $G'-S$, a contradiction.
\end{proof}

\subsection{Analysis of the Kernel}

Let~$(G,k)$ be an irreducible instance with respect to the above reduction rules, and let~$X$ be a modulator of~$(G, k)$. If $|X|>4k$, we can immediately conclude that the instance is a {\noins} (in this case, we return a trivial {\noins}). Assume now that $|X|\leq 4k$. Cygan~et~al.~\cite{DBLP:journals/mst/CyganPPLW17} observed that for every $x\in X$, there can be at most two bags in $\cB(G-X)$ which are attached to~$x$. In fact, if this is not the case there would be an induced claw (with~$x$ being the center and three vertices from three bags attached to~$x$ being the leaves), contradicting the maximality of~$X$. This observation directly offers an upper bound of the number of attached bags.

\begin{lemma}
\label{lem_number_of_D_1_bags}
There are at most~$8k$ attached bags in~$\cB(G-X)$.
\end{lemma}

The next lemma bounds the size of each big bag.

\begin{lemma}
\label{lem_size_bound_of_D_1_bags}
Every big bag in $\cB(G-X)$ contains at most~$8k$ vertices.
\end{lemma}

\begin{proof}
Let~$B$ be a big bag in $\cB(G-X)$. Assume that $|B|\geq \nbg$ (otherwise, we are done). Due to Rule \ref{rule_remove_edges_in_outlier_bags},~$B$ cannot be an outlier bag. Due to Rule \ref{Rule_replacing_D_2},~$B$ cannot be a border bag too. Hence,~$B$ must be an attached bag. Let~$v$ be any arbitrary vertex in~$B$. Due to Lemma~\ref{lem:bag-decomposition},~$v$ belongs to exactly two bags. Let~$B'$ be the other bag including~$v$. If~$B'$ is a border bag, Rule~\ref{rule_remove_one_vertex_from_a_big_bag} applies. Hence,~$B'$ must be an attached bag. Due to Lemma~\ref{lem:bag-decomposition}, every two bags share at most one vertex. As there are at most~$8k$ attached bags in~$\cB(G-X)$ (Lemma \ref{lem_number_of_D_1_bags}) and~$v$ is chosen arbitrarily,~$B$ includes at most~$8k$ vertices.
\end{proof}

Now we are ready to prove Theorem~\ref{thm-kernel}.

\begin{proof}[Proof of Theorem~\ref{thm-kernel}]
The kernelization applies Rules~~\ref{rule_remove_components}--\ref{Rule_replacing_D_2} until none of them is applicable. Notice that each application of a reduction rule, except Rule \ref{Rule_replacing_D_2}, strictly decreases the size of the instance. Hence, Rules~\ref{rule_remove_components}--\ref{rule_remove_one_vertex_from_a_big_bag} can be applied at most polynomial times. Rule \ref{Rule_replacing_D_2} may increase the size of the instance. However, each application of Rule~\ref{Rule_replacing_D_2} destroys one border bag of size at least~$\nbg$. As there can be at most polynomially many such bags (implied by Lemma~\ref{lem-bags-computable-polynomial-time}), Rule~\ref{Rule_replacing_D_2} can be applied only at most polynomial times too. Moreover, as each application of a reduction rule takes polynomial-time, the kernelization terminates in polynomial time.

It remains to compute the size of the kernel. Let $(G, k)$ be the irreducible instance and~$X$ a modulator of~$(G, k)$. If $|X|>4k$, we return a trivial {\noins}. Assume now that $|X|\leq 4k$. Due to Rules~\ref{rule_remove_components}--\ref{rule_remove_edges_in_outlier_bags}, there are no outlier bags. Moreover, due to Lemmas~\ref{lem_number_of_D_1_bags} and~\ref{lem_size_bound_of_D_1_bags}, the number of vertices in attached bags is bounded by $8k\cdot 8k=\bigo(k^2)$. It remains to bound the number of vertices in border bags. Due to Lemma~\ref{lem:bag-decomposition}, every vertex in $G-X$ is in exactly two bags in $\cB(G-X)$. This implies that there are at most~$\bigo(k^2)$ many border bags. Then, due to Rule~\ref{Rule_replacing_D_2} we can conclude that there are at most $\bigo(k^2)\cdot (2k+2)=\bigo(k^3)$ vertices in border bags. In summary,~$|V(G)|$ is bounded by~$\bigo(k^3)$. %Finally, we have \[|V(G)|=4k+\bigo(k^2)+\bigo(k^3)=\bigo(k^3).\]

It remains to analyze the number of edges in~$G$. Clearly, there are at most~$\bigo(k^2)$ edges in~$G[X]$, and at most $4k\cdot \bigo(k^2)=\bigo(k^3)$ edges between~$X$ and attached bags. As there are at most~$8k$ attached bags, and each of them is of size at most~$8k$ (Lemma~\ref{lem_size_bound_of_D_1_bags}), the number of edges between vertices in attached bags is~$\bigo(k^3)$. As discussed above, there are at most~$\bigo(k^2)$ many border bags and each of them is of size at most~$\nsmbg$ (due to Rule~\ref{Rule_replacing_D_2}). Hence, the number of edges between vertices in border bags is bounded by $\bigo(k^2)\cdot \bigo(k^2)=\bigo(k^4)$. According to Condition~\ref{p:nonadj} of Lemma~\ref{lem:bag-decomposition}, there are no other edges. Therefore,~$G$ has at most~$\bigo(k^4)$ edges.
\end{proof}

\section{An {\fpt} Algorithm}
In this section, we study a branching algorithm for {\probshort} to prove Theorem~\ref{thm-algorithm}.
Branching algorithms are commonly used to solve {\nph} optimization problems. In general, a branching algorithm splits (branches) an instance into several subinstances, recursively solves each subinstance, and then combines the solutions of subinstances to a solution of the original instance. A {\it{branching rule}} prescribes how to split the instances. Let~$p$ be a parameter associated with a problem for the purpose of branching (in our case, $p=k$ is the number of edges needed to be deleted). For a branching rule which splits an instance into~$j$ subinstances with new parameters $p-a_1, p-a_2,\dots, p-a_j$,  $ \langle a_1,\dots,a_j \rangle$ is called the {\it{branching vector}} of the branching rule. In addition, the {\it{branching factor}} of the branching rule is the unique positive root of the linear recurrence
\begin{equation}
\label{equ-branching}
x^p-x^{p-a_1}-x^{p-a_2}\dots-x^{p-a_j}=0
\end{equation}
The running time of a branching algorithm is bounded by~$O^*(c^p)$ where~$c$ is the maximum branching factor of all branchings it contains.
%If all possible cases are covered by the branching rules, the branching algorithm correctly solves the problem.
For the reader who is unfamiliar with branching algorithms, we refer to~\cite[Chapter~2]{fominKratschBook2010} for a gentle introduction.

As an induced diamond has five edges and an induced claw has three edges, directly branching on edges in induced claws and diamonds leads to an $\bigo^*(5^k)$-time algorithm. Based on refined observations, we derive branching rules leading to an improved algorithm of worst-case running time~$\bigo^*(3.792^k)$.

The first branching rule is on induced claws, i.e., each subinstance after the rule corresponds to the deletion of one edge in the claw considered at the moment.
%Particularly, given an induced claw, the rule branches the instance into three subinstances each of which is obtained from the original instance by deleting one edge of the claw (and with new parameter $k-1$).
Clearly, the branching factor of this branching rule is~$3$. The algorithm applies the above branching rule once there are any induced claws in the graph. Hence, before branching upon an induced diamond, we always assume there is no induced claws. Now we derive a branching rule on induced diamonds. Let {\fst} be an induced diamond as shown in the figure below.
\begin{center}
\includegraphics[width=0.18\textwidth]{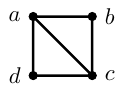}
\end{center}

We distinguish between the following cases. For a collection $\{E_1,E_2,\dots,E_j\}$ of subsets of edges, a branching rule which branches the instance into~$j$ subinstances where the~$i$-th subinstance, $1\leq i\leq j$, is obtained from the original instance by deleting exactly the edges in~$E_i$ and decreasing the parameter~$k$ by~$|E_i|$, is denoted by $\{-E_1,-E_2,\dots,-E_j\}$. Each~$-E_i$ is called a {\it{branching case}} of the branching rule.

Case~1.  If none of the vertices in {\fst} has neighbors outside~{\fst}, we directly delete~{\fst} and decrease~$k$ by one. %Note that in this case we actually do not need to branch.

Case~2. The two vertices~$a$ and~$c$ are twins, i.e., $N_G(a)\setminus \{c\}=N_G(c)\setminus \{a\}$. Then, due to symmetry, it suffices to consider the branching rule $\{-\{\edge{a}{d}\}, -\{\edge{a}{c}\}, -\{\edge{a}{b}\}\}$. The branching vector  and the branching factor of this branching rule are clearly $\langle 1,1,1\rangle$ and~$3$, respectively.

Case~3. There is a vertex~$t$ which is adjacent to~$a$ but not to~$c$.  Then,~$t$ must be adjacent to at least one of~$b$ and~$d$, since otherwise there is an induced claw.  We distinguish between two subcases.

\begin{figure*}[ht!]
\begin{center}
\includegraphics[width=0.95\textwidth]{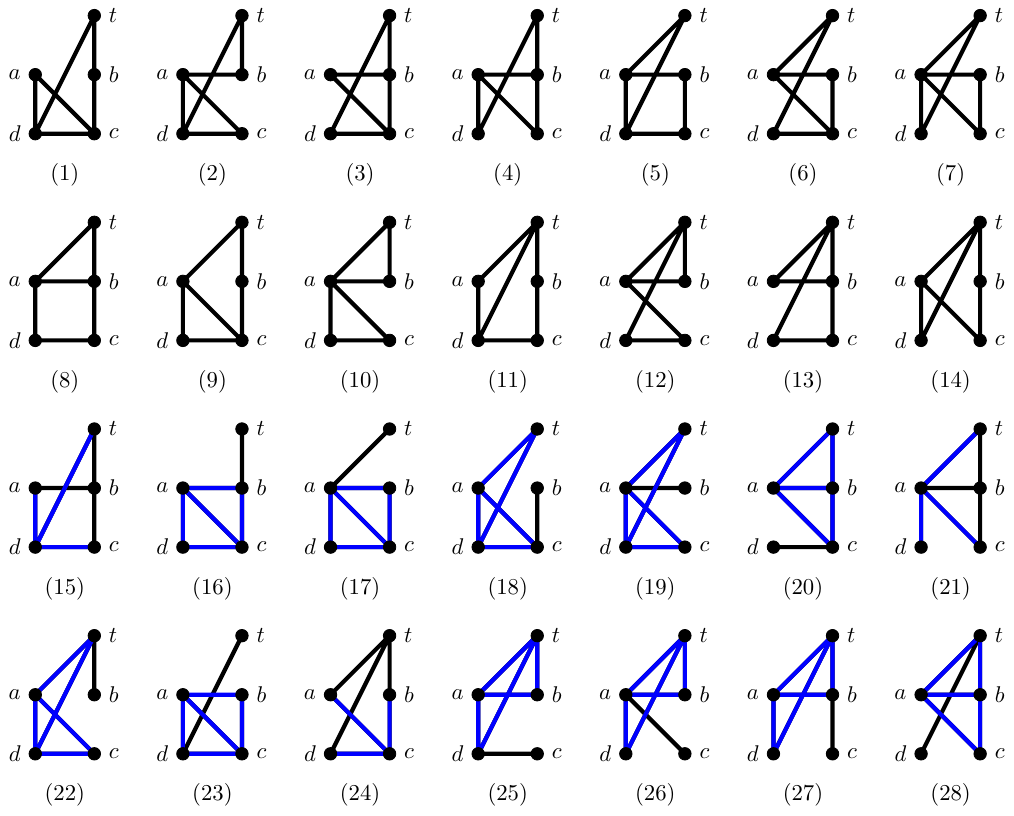}
\end{center}
\caption{All combinations of two edges in the subgraph induced by~$a$,~$b$,~$c$,~$d$, and~$t$. Each figure~($i$), $1\leq i\leq 28$, has two edges in the subgraph induced by~$a$,~$b$,~$c$,~$d$, and~$t$ being deleted.}
\label{fig-many-branchings}
\end{figure*}

Case~3.1. $t$ is adjacent to exactly one of~$d$ and~$b$. Without loss of generality, assume that~$t$ is adjacent to~$b$. Clearly,~$a$,~$t$,~$b$, and~$c$ also induce a diamond which shares the edges~$\edge{a}{b}$,~$\edge{b}{c}$, and $\edge{a}{c}$ with {\fst}. We first branch on deleting each of these three edges (i.e., the three branching cases $-\{\edge{a}{b}\}$, $-\{\edge{b}{c}\}$ and $-\{\edge{a}{c}\}$). Consider the remaining branching cases, i.e., none of~$\edge{a}{b}$,~$\edge{b}{c}$ and~$\edge{a}{c}$ is deleted. Observe that in this case we have to delete at least two edges in order to destroy~{\fst} and the induced diamond formed by~$a$,~$t$,~$b$ and,~$c$. There are in total four branching cases to consider:
\begin{enumerate}
\item  $-\{\edge{a}{d}, \edge{a}{t}\}$;
\item  $-\{\edge{a}{d}, \edge{b}{t}\}$;
\item  $-\{\edge{c}{d}, \edge{a}{t}\}$;  and
\item  $-\{\edge{c}{d}, \edge{b}{t}\}$.
\end{enumerate}
However, observe that after deleting the edges~$\edge{c}{d}$ and~$\edge{b}{t}$, the set $\{a,d,b,t\}$ induces a claw with~$a$ being the center. This implies that we need to delete at least one edges in $\{\edge{a}{d}, \edge{a}{b}, \edge{a}{t}\}$ to destroy the induced claw. In other words, the branching case  $-\{\edge{c}{d}, \edge{b}{t}\}$ has been covered by other cases and hence can be discarded. In summary, we have a branching vector $\langle 1,1,1,2,2,2\rangle$. By solving Equation~(\ref{equ-branching}), we obtain a branching factor~$3.792$.

Case~3.2. $t$ is adjacent to both~$d$ and~$b$. In this case, there are four induced diamonds in the graph induced by~$a$,~$b$,~$c$,~$d$ and,~$t$ (except $\{b,c,d,t\}$, all other $4$-subsets of $\{a,b,c,d,t\}$ induce diamonds). More importantly, at least two edges have to be deleted in order to destroy these four induced diamonds. As there are eight edges in the subgraph induced by $\{a,b,c,d,t\}$, there are in total ${8\choose 2}=28$ cases to consider. \myfig{fig-many-branchings} shows all these~$28$ cases, with the missing edges in the subgraph induced by $\{a,b,c,d,t\}$ being the deleted edges in each case. However, we claim that we need only to consider branching cases (1)--(14). The reason is that in any other case there is still an induced claw or diamond after deleting the corresponding two edges (see the subgraph with blue edges in each case). In order to destroy these induced claws or diamonds, further edges must be deleted. Therefore, each case ($i$) where $15\leq i\leq 28$ is covered by some of the cases (1)--(14). For instance, in Case~(15) (i.e., after deleting the edges $\edge{a}{t}$ and $\edge{a}{c}$), $\{d,a,c,t\}$ induces a claw. To destroy this claw, we need further delete one of the edges in the claw. Clearly, deleting further~$\edge{a}{d}$ is covered by Case~(3), deleting $\edge{d}{t}$ is covered by Case~(8), and deleting~$\edge{d}{c}$ is covered by Case~(4). In summary, in Case~3.2 we have~$14$ branching cases to consider (branching Cases (1)--(14)). As each branching case decreases the parameter~$k$ by two, the corresponding branching factor is the unique positive root of~$x^2=14$ (see Equation~(\ref{equ-branching})), which is~$3.742$.

Case~4. There is a vertex~$t$ which is adjacent to~$c$ but not to~$a$. This case is symmetric to Case~3 and can be dealt with similarly.

Clearly, Case~3 has a branching rule with the maximum branching factor~$3.792$. Hence, the algorithm has worst-case running time~$\bigo^*(3.792^k)$, completing the proof of Theorem~\ref{thm-algorithm}.

\section{Conclusion}
In this paper, we have investigated the kernelization and {\fpt}-algorithm of {\probshort}. In particular, based on five reduction rules, we obtained a kernel of~$\bigo(k^3)$ vertices and~$\bigo(k^4)$ edges, significantly improving the previous kernel with~$\bigo(k^{12})$ vertices and~$\bigo(k^{24})$ edges. In addition, based on refined observations, we devised an {\fpt}-algorithm of running time~$\bigo^*(3.792^k)$. A natural direction for future research could be to investigate whether {\probshort} admits a square vertex kernel. 
%In addition, exploring approximation algorithms for {\probshort} is also an intriguing question.

%\bibliographystyle{plain}
%%\bibliographystyle{plainurl}
%\bibliography{../Bibs/graphref,../Bibs/sociachoiceref}

\begin{thebibliography}{10}

\bibitem{DBLP:journals/algorithmica/BetzlerBBNU14}
N.~Betzler, H.~L. Bodlaender, R.~Bredereck, R.~Niedermeier, and J.~Uhlmann.
\newblock On making a distinguished vertex of minimum degree by vertex
  deletion.
\newblock {\em Algorithmica}, 68(3):715--738, 2014.

\bibitem{DBLP:journals/ipl/Cai96}
L.~Cai.
\newblock Fixed-parameter tractability of graph modification problems for
  hereditary properties.
\newblock {\em Inf. Process. Lett.}, 58(4):171--176, 1996.


%\bibitem{DBLP:conf/esa/00010SY18}
%Y.~Cao, A.~Rai, R.~B.~Sandeep, and J.~Ye. 
%\newblock A polynomial kernel for diamond-free editing.
%\newblock In {\em Proceedings of the 26th Annual European Symposium on Algorithms}, {LIPIcs} {112}, Nr.~13, 2018.

\bibitem{DBLP:journals/corr/abs-2001-06867}
C.~Crespelle, P.~G. Drange, F.~V. Fomin, and P.~A. Golovach.
\newblock A survey of parameterized algorithms and the complexity of edge  modification.
\newblock {\em CoRR}, abs/2001.06867, 2020.

\bibitem{DBLP:journals/algorithmica/CyganMPPS14}
M.~Cygan, D.~Marx, M.~Pilipczuk, M.~Pilipczuk, and I.~Schlotter.
\newblock Parameterized complexity of {E}ulerian deletion problems.
\newblock {\em Algorithmica}, 68(1):41--61, 2014.

\bibitem{DBLP:journals/mst/CyganPPLW17}
M.~Cygan, M.~Pilipczuk, M.~Pilipczuk, E.~J. van Leeuwen, and M.~Wrochna.
\newblock Polynomial kernelization for removing induced claws and diamonds.
\newblock {\em Theory. Comput. Syst.}, 60(4):615--636, 2017.

\bibitem{DBLP:series/txcs/DowneyF13}
R.~G. Downey and M.~R. Fellows.
\newblock {\em Fundamentals of Parameterized Complexity}.
\newblock Texts in Computer Science. Springer, 2013.

\bibitem{fominKratschBook2010}
F.~V. Fomin and D.~Kratsch.
\newblock {\em Exact Exponential Algorithms}, chapter~2, pages 13--30.
\newblock Texts in Theoretical Computer Science An EATCS Series. Springer,
  2010.

\bibitem{DBLP:conf/wg/KomusiewiczNN15}
C.~Komusiewicz, A.~Nichterlein, and R.~Niedermeier.
\newblock Parameterized algorithmics for graph modification problems: On  interactions with heuristics.
\newblock In {\em Proceedings of the 41st International Workshop on Graph-Theoretic Concepts in Computer Science}, Lecture Notes in Computer Science 9224, pages 3--15, 2015.

\bibitem{DBLP:journals/tcs/LiuWYCC15}
Y.~Liu, J.~Wang, J.~You, J.~Chen, and Y.~Cao.
\newblock Edge deletion problems: Branching facilitated by modular
  decomposition.
\newblock {\em Theor. Comput. Sci.}, 573:63--70, 2015.

\bibitem{DBLP:journals/algorithmica/MarxS12}
D.~Marx and I.~Schlotter.
\newblock Obtaining a planar graph by vertex deletion.
\newblock {\em Algorithmica}, 62(3-4):807--822, 2012.

\bibitem{DBLP:conf/iwpec/SandeepS15}
R.~B. Sandeep and N.~Sivadasan.
\newblock Parameterized lower bound and improved kernel for diamond-free edge
  deletion.
\newblock In {\em Proceedings of the 10th International Symposium on Parameterized and Exact Computation}, Leibniz International Proceedings in Informatics 43, pages 365--376, 2015.

\bibitem{DBLP:journals/jda/Sritharan16}
R.~Sritharan.
\newblock Graph modification problem for some classes of graphs.
\newblock {\em J. Discrete Algorithms}, 38-41:32--37, 2016.

\bibitem{DBLP:journals/corr/abs-1908-07318}
D.~Tsur.
\newblock An algorithm for destroying claws and diamonds.
\newblock {\em CoRR}, abs/1908.07318, 2019.

\end{thebibliography}

\end{document}